\documentclass[margin=0.4in]{article}

\usepackage[utf8]{inputenc}
\usepackage{float}
\usepackage{geometry}
\usepackage{amsthm}
\usepackage[all]{xy}
\usepackage{amsmath}
\usepackage{amssymb}
\usepackage{amsfonts}
\usepackage{ebproof}
\usepackage{graphicx}
\usepackage{xcolor}

\usepackage{enumerate}

\newcommand{\tset}[1]{[\![#1]\!]}

\newcommand{\winform}[3]{#1 \; \begin{array}{| c |} #2 \\ #3 \end{array} \;}

\newcommand{\moncut}[1]{$\mathsf{cut}$}

\newcommand{\caseop}[3]{[#2\;\mathsf{if}\; #1 \;\mathsf{else}\; #3] }

\newcommand{\term}[1]{\mathbf{#1}}

\NewDocumentCommand{\termSum}{O{}}{ \term{\sumSym} {#1} }
\NewDocumentCommand{\termPlus}{m O{}}{ \term{\plusSym}\IfValueT{#2}{\mathrlap{^{\!#2}}}_{\!#1} }
\NewDocumentCommand{\termTimes}{o}{ \term{\prodSym} \IfValueT{#1}{^{#1\kern-1pt}} }
\NewDocumentCommand{\termResid}{o}{ \term{\resSym} \IfValueT{#1}{^{#1}} }

\NewDocumentEnvironment{scriptdisplay}{}
    {\begin{equation*}\scriptstyle}
    {\end{equation*}}

\NewDocumentEnvironment{ssprooftree}{}
    {\begin{prooftree}[compact]}
    {\end{prooftree}}

\NewDocumentEnvironment{ssprooftree*}{}
    {\begin{prooftree*}[compact]}
    {\end{prooftree*}}

\newcommand{\ReductionArrow}{\Rightarrow}
\DeclareDocumentCommand{\reduction}{ O{} m m }{%
    \begin{prooftree}[#1] #2 \end{prooftree}%
    \;\; \ReductionArrow \;\;%
    \begin{prooftree}[#1] #3 \end{prooftree}%
  }

\NewDocumentCommand \ntStyle { m }{ \mathsf{#1} }
\NewDocumentCommand \nt { s m m }{ \IfBooleanTF{#1}{\hat{\ntStyle{N}}}{\ntStyle{N}}^{#2}_{#3} }

\NewDocumentCommand{\setof}{ m o }{\{ \, #1 \IfValueT{#2}{ : #2}\, \} }
\NewDocumentCommand{\Setof}{ m o }{\Bigl\{ #1 \IfValueT{#2}{ : #2} \Bigr\} }
\NewDocumentCommand{\SETOF}{ m o }{\Biggl\{ #1 \IfValueT{#2}{ : #2} \Biggr\} }

\NewDocumentCommand{\Sub}{mo}{[ #1\IfNoValueF{#2}{/#2} ]}

\RenewDocumentCommand{\perp}{ m }{ \neg {#1} }
\newcommand{\plusSym}{{+}}

\newcommand{\prodSym}{{\times}}

\newcommand{\sumSym}{}

\newcommand{\resSym}{{/}}

\NewDocumentCommand{\rul}{ m e{_} o }{%
    \ensuremath{\mathsf{#1}%
    \IfValueT{#2}{_{#2}}%
    \IfValueT{#3}{^{#3}}
    }}
\NewDocumentCommand{\nrul}{ m o }{%
    \ensuremath{\perp{\mathsf{#1}}%
    \IfValueT{#2}{^{#2}}%
    }}
\NewDocumentCommand{\ruleset}{ o }{\operatorname{\mathit{Rules}}\IfValueT{#1}{(#1)}}
\NewDocumentCommand{\var}{o}{\mathit{V}\IfValueT{#1}{\!(#1)}}

\newtheorem{theo}{Theorem}

\newtheorem{prop}{Proposition}

\theoremstyle{definition}
\newtheorem{defi}{Definition}
\theoremstyle{definition}

\theoremstyle{definition}

\theoremstyle{definition}

\newcommand{\QFS}{\mathsf{QF}}
\newcommand{\qfsval}{\vdash_{\QFS}}

\title{From Herbrand schemes to functional interpretation}

\author{Sebastian Enqvist-Pyk}

\begin{document}
\maketitle

\begin{abstract}
Herbrand schemes are a method to extract Herband disjunctions directly from sequent calculus proofs, without appealing to cut elimination, using a formal grammar known as a higher-order recursion scheme. In this note, we show that the core ideas of Herbrand schemes can be reformulated as a functional interpretation of classical sequent calculus, similar to the functional interpretation of classical logic due to Gerhardy and Kohlenbach. We argue that this provides a natural computational interpretation of classical sequent calculus, in the same spirit as the game-theoretic approach due to Alcolei et al. that has previously been used to analyze Herbrand's theorem in terms of concurrency.  
\end{abstract}

\section{Introduction}

The conceptual connection between Herbrand's theorem and functional interpretation is well known. It was made explicit by Gerhardy and Kohlenbach in \cite{GerhardyK05}, in which Herbrand's theorem is proved using a functional interpretation of the classical first-order predicate calculus. The approach to Herbrand extraction via functional interpretation was further explored by Ferreira and Ferreira \cite{ferreira2017herbrandized}, who developed a so-called ``Herbrandized'' functional interpretation for classical first-order logic (see also \cite{van2012functional}).  In both these works, we obtain the result that the realizer extracted from a proof of an existential formula, when reduced to normal form, describes the witnessing set of terms for a Herbrand disjunction. Thus, Herbrand's theorem is elegantly located as a special case of functional interpretation. 

These works form part of a quite large literature on various approaches to Herbrand's theorem, using a variety of tools including proof nets \cite{heijltjes2010classical,mckinley2013proof},  expansion trees \cite{miller1987compact,aschieri2019expansion} and game semantics \cite{alcolei2018true}.  The present work builds on an approach developed by Afshari, Hetzl and Leigh \cite{APAL20}, and further in \cite{afshari2025herbrand,herbrandcyclic}, in which Herbrand disjunctions are extracted directly from sequent calculus proofs without relying on cut elimination. Functional interpretations, by contrast, are usually carried out for Hilbert-style axiomatic systems. 

The idea of Herbrand schemes is to view a sequent calculus proof as a type of grammar, known as a \emph{higher-order recursion scheme}. These grammars associate certain non-terminal symbols with formula occurrences in the end sequents of proofs, and these are supplied with rewrite rules depending on the last inference rule of the proof. For a proof of an existential formula, the language generated by the grammar is a set of witnessing terms for a Herbrand disjunction.  Again, there is a connection with functional interpretation here, noted already in \cite{APAL20}, in the use of higher-order types.

In this note we aim to further complete the emerging picture by showing how, using ideas from Herbrand schemes, we can derive realizing terms in a typed lambda-calculus similar to that used by Gerhardy and Kohlenbach directly from proofs in classical sequent calculus. These realizers allow us to construct Herbrand disjunctions for valid existential formulas. We thus obtain again a method to extract Herbrand disjunctions directly from sequent calculus proofs without relying on cut elimination, but this time using functional interpretation instead of a higher-order grammar.

A full presentation of Herbrand schemes would take up too much space here; the interested reader is advised to consult \cite{APAL20,afshari2025herbrand,herbrandcyclic} for comparison. The results presented here are self-contained and do not require a previous understanding of Herbrand schemes. However, the procedure to extract realizing terms from proofs presented here was found by a close examination of the type structure and rewrite rules given in \cite{APAL20}. Indeed the way that term extraction works here remains \emph{very} closely related to Herbrand schemes. In particular, we retain the following two features:

\begin{enumerate}
\item  Each formula $A$ is associated with two distinct types: an \emph{evidence type} $[A]$ and a \emph{counter-evidence type} $\langle A \rangle$ (called ``output'' and ``input'' types in \cite{APAL20}).
\item For each proof, and each formula occurrence in the end sequent, we essentially construct a program that extracts evidence for that formula from potential counter-evidence for each formula occurence of the end sequent. 
\end{enumerate}
This program extraction is modular and pieces together programs extracted from immediate subproofs uniformly from the last inference rule used. In Herbrand schemes, these programs are represented by non-terminal symbols of the grammar with associated rewrite rules. If a sequent calculus proof $p$ has end sequent $A_1,\hdots,A_k$, then a Herbrand scheme associates with each index $i$ a non-terminal symbol $F^p_i$ of type:
\[ \Sigma \to \langle A_1 \rangle \to \dotsm \to \langle  A_n \rangle \to [A_i] \]
where $\langle A_1 \rangle, \hdots, \langle A_n \rangle$ are the ``counter-evidence types'' of $A_1,\hdots, A_n$ and $[A_i]$ is the evidence type of the $i$-th formula $A_i$, and $\Sigma$ is a distinguished type of ``substitution stacks'' (which will not be used here). Here, we will similarly associate with each proof of end sequent  $A_1,\hdots, A_n$ and index $i$ a term of type $[A_i]$ parametric in given terms of types $\langle A_1 \rangle,\hdots,\langle A_n \rangle$; we will call this assignment a \emph{term constructor}.
In contrast with Herbrand schemes, however, there is no need to specify particular ad hoc rewrite rules for these terms; once the term has been constructed, its behaviour follows from the basic rewrite rules of the type theory.

The approach of assigning explicit ``evidence'' and ``counter-evidence'' types to formulas has appeared before in connection with functional interpretation, in Pédrot's ``functional functional interpretation'' \cite{pedrot2014functional} which build's on de Paiva's category-theoretical treatment of functional interpretation \cite{de1991dialectica}. As far as we are aware, this approach has not previously been developed for classical logic and classical sequent calculus.

We believe this work provides further evidence of how natural the connection between functional interpretation and Herbrand's theorem is: it turns out that Herbrand schemes can be reformulated as a functional interpretation of sequent calculus proofs. Furthermore, the functional interpretation described here has a conceptual connection with the game-theoretic analysis of Herbrand's theorem described in \cite{alcolei2018true}, and provides a computational interpretation of sequent calculus as concurrent computation in the same spirit.  We elaborate on this point in Section \ref{s:concurrency}.

\section{Types and terms}

We begin by presenting the underlying type theory of our functional interpretation. 
We fix a first-order vocabulary $\mathcal{L}$ and assume that it contains at least one individual constant $c$ and at least one unary predicate $P$, and we assume an infinite supply of individual variables.  The set of $\mathcal{L}$-terms is formed in the usual way, where terms are formed from individual constants, variables and function symbols.  

Our types are generated by the following grammar:
\[U : = \iota \mid \Box  \mid U \times U \mid U \to U\]
We thus have two atomic types, a type $\iota$ of individuals and a ``null type'' $\Box$. These appear also in \cite{APAL20,afshari2025herbrand,herbrandcyclic}.  We will equate individual variables with variables of type $\iota$. Complex types are formed by products and function space types. 

Terms are constructed and typed according to the following rules, where we define the set $\mathsf{Prop}$ of well-formed propositions by mutual recursion below:

\begin{figure}[H]
\center
\begin{tabular}{ c c }
\begin{prooftree}
\hypo{}
\infer1{\varepsilon : \Box}
\end{prooftree}
&
\begin{prooftree}
\hypo{t \text{ an $\mathcal{L}$-term}}
\infer1{t  : \iota}
\end{prooftree}
\\[1.5em]
\begin{prooftree}
\hypo{s : U}
\hypo{t : V}
\infer2{(s, t) : U \times V}
\end{prooftree}
&
\begin{prooftree}
\hypo{s : U_1 \times U_2}
\infer1{\pi_i(s) : U_i}
\end{prooftree}
\\[1.5em]
\begin{prooftree}
\hypo{x : U}
\hypo{t : V}
\infer2{\lambda x. t : U \to V}
\end{prooftree}
&
\begin{prooftree}
\hypo{s : U \to V}
\hypo{t : U}
\infer2{st : V}
\end{prooftree}
\end{tabular}
\[
\begin{prooftree}
\hypo{u : U}
\hypo{v : U}
\hypo{A \in \mathsf{Prop}}
\infer3{\caseop{A}{u}{v} : U}
\end{prooftree}
\]
\end{figure}

We refer to terms that can be typed by these rules as \emph{$\mathcal{L}^+$-terms.} Free variables of terms are defined as usual, and a \emph{closed} term is one with no free variables.
The operator $\caseop{A}{*}{*}$  is the case distinction operator from Gerhardy and Kohlenbach \cite{GerhardyK05}, and uses a proposition as parameter. The set of propositions is defined as follows:

Propositions:
\[
\begin{prooftree}
\hypo{A \in \mathsf{Prop}}
\hypo{B \in \mathsf{Prop}}
\infer2{A \vee B \in \mathsf{Prop}}
\end{prooftree}
\qquad\qquad 
\begin{prooftree}
\hypo{A \in \mathsf{Prop}}
\infer1{\neg A \in \mathsf{Prop}}
\end{prooftree}
\]
\[
\begin{prooftree}
\hypo{t_1 :\iota}
\hypo{\dotsm}
\hypo{t_n : \iota}
\infer3{R(t_1,\hdots,t_n) \in \mathsf{Prop}}
\end{prooftree}
\quad\quad
\begin{prooftree}
\hypo{u_1 : U}
\hypo{u_2 : U}
\infer2{u_1 \equiv u_2  \in \mathsf{Prop}}
\end{prooftree}
\]

A first observation about our type theory is the following:

\begin{prop}
\label{p:inhabited}
Every type is inhabited by some closed term. 
\end{prop}
\begin{proof}
By a straightforward induction on the complexity of a type $U$: the type $\Box$ is inhabited by $\varepsilon$ and $\iota$ is inhabited by the constant $c$. Assuming $U, V$ are inhabited by $u,v$ the type $U \times V$ is inhabited by $(u,v)$ and the type $U \to V$ is inhabited by $\lambda z. v$.
\end{proof}

To give operational meaning to our terms we take the following rewrite rules:

\begin{figure}[H]
\center
\begin{tabular}{c}
$(\lambda x. u)v \longrightarrow u[v/x]$ 
\\[1.5em]
$\pi_1(u,v) \longrightarrow u$
\\[1.5em]
$\pi_2(u,v) \longrightarrow v$
\\[1.5em] 
$f(\caseop{A}{u}{v}) \longrightarrow \caseop{A}{fu}{fv}$ 
\\[1.5em]
$(\caseop{A}{f}{g})u \longrightarrow \caseop{A}{fu}{gu}$
\end{tabular}
\end{figure}

The rewrite rule for $\lambda$-abstractions will be referred to as $\beta$-reduction as usual.

\begin{prop}
\label{p:strong-normalization}
The reduction relation $\longrightarrow$ generated by these rules is strongly normalizing. Furthermore, if $t$ is a closed normal form term of type $U$, then either $t$ has the form $\caseop{A}{u}{v}$ for some $u,v$, or it has the form:
\begin{center}
\begin{tabular}{c  c c}
$\varepsilon$ & \text{ if } & $U = \Box$ \\
an $\mathcal{L}$-term & \text{ if } & $U = \iota$ \\
$(v,w)$ & \text{ if } & $U = V \times W$ \\
$\lambda x. w$ & \text{ if } & $U = V \to W$ \\
\end{tabular}
\end{center}
\end{prop}

\begin{proof}
Given a type $U$ we define the set of terms $\tset{U}$ by the following recursion:
\\\\
- $u \in \tset{\iota}$ iff $u$ strongly normalizes and every normal form of $u$ is either an $\mathcal{L}$-term or of the form $\caseop{A}{v}{w}$. 
\\\\
- $u \in \tset{\Box}$ iff $u$ strongly normalizes and every normal form of $u$ is either equal to $\varepsilon$ or of the form $\caseop{A}{v}{w}$.
\\\\
 - $u \in \tset{V \times W}$ iff $\pi_1(u) \in \tset{V}$ and $\pi_2(u) \in \tset{W}$.
\\\\
- $u \in \tset{V \to W}$ iff, for all $v \in \tset{V}$, we have $uv \in  \tset{W}$.
\\\\
A standard induction on types shows that every term in $\tset{U}$ strongly normalizes, and by induction on complexity of terms we can show that every term of type $U$ belongs to $\tset{U}$. 

We verify the shape of normal forms, by induction on types. The cases for types $\iota$ or $\Box$ are immediate from the definitions of $\tset{\iota}$ and $\tset{\Box}$. For a normal form term $u : V \times W$, since $u \in \tset{V \times W}$ we have $\pi_1(u) \in \tset{V}$ and $\pi_2(u) \in \tset{W}$. Since the induction hypothesis on types $V$,$W$ entails that neither $\pi_1(u)$ nor $\pi_2(u)$ are in normal form, at least one reduction must be possible for each of these terms. But there are only three ways that we can reduce the term $\pi_1(u)$: either as $\pi_1(u) \longrightarrow \pi_1(u')$ where $u \longrightarrow u'$, or using the reduction $\pi_1(v',w') \longrightarrow v'$ where $u$ is $(v',w')$, or $\pi_1(\caseop{A}{v'}{w'}) \longrightarrow \caseop{A}{\pi_1(v')}{\pi_1(w')}$. The first case is excluded by the assumption that $u$ was in normal form. In the second case $u$ ia of the form $(v',w')$, and in the third case $u$ is of the form $\caseop{A}{v'}{w'}$, in line with the stated shape of normal forms. For the case where $u : V \to W$, pick an arbitrary inhabitant $v : V$ in normal form and \emph{not} of the form $\caseop{A}{w}{w'}$; such inhabitants exist for every type. We have $v \in \tset{V}$ and $u \in \tset{V \to W}$. The induction hypothesis on $W$ entails that $uv$ is not in normal form, so at least one reduction must be possible. Since both $u$ and $v$ were assumed to be in normal form, and $v$ is not of the form   $\caseop{A}{w}{w'}$, the only ways to reduce $uv$ is by a $\beta$-reduction $uv \longrightarrow w$, or by a reduction $\caseop{A}{f}{g}v \longrightarrow \caseop{A}{fv}{gv}$. The first case is only possible if $u$ is a $\lambda$-abstraction, and the second one only if $u$ is of the form $\caseop{A}{f}{g}$. In both cases $u$ satisfies the required shape of normal forms.  
\end{proof}


We introduce the propositional connectives $\to, \wedge, \leftrightarrow$ by the usual abbreviations. We shall define a quantifier-free calculus $\QFS$ for reasoning with propositions as follows.
As axioms we take every instance of a propositional tautology, every equation $u \equiv v$ for convertible terms $u,v$, plus the following axioms:

\[A \to \caseop{A}{u}{v} \equiv u \qquad \neg A \to \caseop{A}{u}{v} \equiv v\]
 As rules of inference, we take substitution and modus ponens:

\[ 
\begin{prooftree}
\hypo{A[u/x]}
\hypo{u \equiv v}
\infer2{A[v/x]}
\end{prooftree}
\qquad
\begin{prooftree}
\hypo{A}
\hypo{A \to B}
\infer2{B}
\end{prooftree}
\]

 We write $\qfsval A $ if $A$ is provable in $\QFS$, and sometimes write $A_1,\hdots,A_n \qfsval B$ to abbreviate the statement $\qfsval A_1 \wedge \hdots \wedge  A_n \to B$.  We write $A \equiv B$, overloading notation, if $\qfsval A \leftrightarrow B$.

\subsection{Semantics}

Given a set $D$ we define the full set theoretic type structure $\{D_U\}_{U \in \mathsf{Types}}$ on $D$ by induction on types as usual, so that $D_{U \to V}$ is the set of all functions $D_U \to D_V$, etc.

Given a first order structure $M = (D,V)$, where $D$ is the domain and $V$ is the valuation of constants, function symbols and predicates, an \emph{assignment} $a$ is a map sending a variable $z : U$ to an element of $D_U$. We extend the assignment $a$ to a map from arbitrary terms $t : U$ to elements $a(t) \in D_U$, and simultaneously define a satisfaction relation $M,a \vDash A$ for $A \in \mathsf{Prop}$, by the following clauses: 
\begin{itemize}
\item $a(c) = V(c)$ for an individual constant $c$,
\item $a(f) = V(f)$ for a function symbol $f$,
\item $a(uv) = $ the unique $d \in D_V$ such that $(a(v),d) \in a(u)$, where $u : U \to V$ and $v : U$,
\item $a(\lambda x.u) = \{(d, a[d/x](u)) \mid d \in D_U\}$ where $u : V$, $x : U$ and the assignment $a[d/x]$ is like $a$ except $a[d/x](x) = d$,
\item $a(\caseop{A}{u}{v}) = a(u)$ if $M,a \vDash A$, otherwise $a(\caseop{A}{u}{v}) = a(v)$, 
\item $M, a \vDash R(u_1,\hdots,u_n)$ if $(a(u_1),\hdots,a(u_n)) \in V(R)$,
\item $M, a \vDash u \equiv v$ if $a(u) = a(v)$,
\item standard clauses for $\vee, \neg$. 
\end{itemize}

We can now state a soundness result for the quantifier free system:

\begin{prop}
\label{p:qfs-soundness}
If $\qfsval A$ then $M,a \vDash A$ for every model $M$ and every assignment $a$. 
\end{prop}
\begin{proof}
Standard induction.
\end{proof}

\section{Sequent calculus}

\subsection{Two-sided sequent calculus}

Given our vocabulary $\mathcal{L}$ we define $\mathcal{L}$-formulas and closed $\mathcal{L}$-sentences as usual, where as basic connectives we take disjunction, negation and the existential quantifier. As our proof system we shall take Gentzen's $\mathbf{G1c}$ sequent calculus as presented in \cite{troelstra2000basic}, but restricted to our basic connectives.  

Sequents are pairs of multisets of formulas. 
As our single axiom we take:
\[
\begin{prooftree}
\hypo{}
\infer1[$\mathsf{id}$]{A \vdash A}
\end{prooftree}
\]
The remaing rules are given as follows:

\begin{figure}[H]
\center
\begin{tabular}{ c c }
\begin{prooftree} 
\hypo{\Gamma \vdash A_i, \Delta}
\infer1[$\vee_R$]{\Gamma \vdash A_1 \vee A_2, \Delta}
\end{prooftree}
& 
\begin{prooftree}
\hypo{\Gamma, A \vdash \Delta}
\hypo{\Gamma, B \vdash \Delta}
\infer2[$\vee_L$]{\Gamma, A \vee B \vdash \Delta}
\end{prooftree}
\\[1.5em]
\begin{prooftree}
\hypo{\Gamma, A \vdash \Delta}
\infer1[$\neg_R$]{\Gamma \vdash \neg A, \Delta}
\end{prooftree}
&
\begin{prooftree}
\hypo{\Gamma \vdash A, \Delta}
\infer1[$\neg_L$]{\Gamma, \neg A \vdash \Delta}
\end{prooftree}
\\[1.5em]
\begin{prooftree}
\hypo{\Gamma \vdash A[t/x], \Delta}
\infer1[$\exists_R$]{\Gamma \vdash \exists x A, \Delta}
\end{prooftree}
&
\begin{prooftree}
\hypo{\Gamma, A[\alpha/x] \vdash \Delta}
\infer1[$\exists_L$]{\Gamma, \exists x A \vdash \Delta}
\end{prooftree}
\\[1.5em]
\begin{prooftree}
\hypo{\Gamma \vdash A,A,\Delta}
\infer1[$\mathsf{c}_R$]{\Gamma \vdash A, \Delta}
\end{prooftree}
&
\begin{prooftree}
\hypo{\Gamma, A,A \vdash \Delta}
\infer1[$\mathsf{c}_L$]{\Gamma, A \vdash \Delta}
\end{prooftree}
\\[1.5em]
\begin{prooftree}
\hypo{\Gamma \vdash \Delta}
\infer1[$\mathsf{w}_R$]{\Gamma \vdash A, \Delta}
\end{prooftree}
&
\begin{prooftree}
\hypo{\Gamma \vdash \Delta}
\infer1[$\mathsf{w}_L$]{\Gamma, A \vdash \Delta}
\end{prooftree}
\end{tabular}

\[\begin{prooftree}
\hypo{\Gamma_1 \vdash A, \Delta_1}
\hypo{\Gamma_2,A \vdash \Delta_2}
\infer2[$\mathsf{cut}$]{\Gamma_1,\Gamma_2 \vdash \Delta_1,\Delta_2}
\end{prooftree}\]

\end{figure}

In the rule $\exists_L$, the eigenvariable $\alpha$ is subject to the usual condition that it must not appear in the conclusion.  

\subsection{One-sided sequent calculus}

The aim here is to extract a functional interpretation directly from sequent calculus proofs, so we want as much as possible to take proofs as they are and avoid proof-theoretic pre-processing. We will only perform a rather trivial proof translation in order to simplify notational matters a bit, and turn the two-sided sequent calculus $\mathbf{G1c}$ into a one-sided sequent calculus. The usual approach to one-sided sequent calculus is to present formulas in negation-normal form, and regard negation as a recursively defined operation on formulas. However, it is implicit in that approach that left and right formula occurrences can be treated in a completely dual manner, and it turns out that we need an asymmetric treatment of left and right formula occurrences here. To retain this asymmetry, we will present the one-sided sequent calculus with an explicit negation. 

Sequents are single multisets of formulas. 
As our single axiom we take:
\[
\begin{prooftree}
\hypo{}
\infer1[$\mathsf{lem}$]{\neg A, A}
\end{prooftree}
\]
The remaining rules are given as follows:
\begin{figure}[H]
\center
\begin{tabular}{ c c }
\begin{prooftree} 
\hypo{\Gamma,  A_i}
\infer1[$\vee$]{\Gamma, A_1 \vee A_2}
\end{prooftree}
& 
\begin{prooftree}
\hypo{\Gamma, \neg A}
\hypo{\Delta, \neg B}
\infer2[$\neg \vee$]{\Gamma, \Delta, \neg (A \vee B)}
\end{prooftree}
\\[1.5em]
\begin{prooftree}
\hypo{\Gamma, A[t/x]}
\infer1[$\exists$]{\Gamma, \exists x A}
\end{prooftree}
&
\begin{prooftree}
\hypo{\Gamma, \neg A[\alpha/x]}
\infer1[$\neg \exists$]{\Gamma, \neg \exists x A}
\end{prooftree}
\\[1.5em]
\begin{prooftree}
\hypo{\Gamma, A,A}
\infer1[$\mathsf{c}$]{\Gamma, A}
\end{prooftree}
&
\begin{prooftree}
\hypo{\Gamma}
\infer1[$\mathsf{w}$]{\Gamma, A}
\end{prooftree}
\\[1.5em]
\begin{prooftree}
\hypo{\Gamma, A}
\infer1[$\neg\neg$]{\Gamma, \neg \neg A}
\end{prooftree}
&
\begin{prooftree}
\hypo{\Gamma, A}
\hypo{\Delta, \neg A}
\infer2[$\mathsf{cut}$]{\Gamma,\Delta}
\end{prooftree}
\end{tabular}
\end{figure}

In the $\exists$-rule, the term $t$ occurring in the premiss is an $\mathcal{L}$-term. We can also consider an extended proof system that allows arbitrary $\mathcal{L}^+$-terms of type $\iota$ to occur in formulas, and to occur in the premiss of an $\exists$-rule; we shall occasionally do this when convenient. 

It is easy to translate any $\mathbf{G1c}$-proof of the sequent $A_1,\hdots, A_k \vdash B_1,\hdots, B_m$ to a one-sided proof of the sequent $\neg A_1,\hdots, \neg A_k, B_1,\hdots, B_m$. The only slight catch is that the $\vee$-rules in $\mathbf{G1c}$ are presented as additive rules; again, it will be convenient here to break the symmetry and interpret the left rule as multiplicative while the right rule is treated as additive. This means that in the translation, we have to insert a series of contractions when we interpret the left $\vee$-rule. 

If $\Gamma$ is $A_1,\hdots,A_n$ then we abbreviate $\neg A_1,\hdots,\neg A_n$ by $\neg \Gamma$. The translation $\tau$ is defined recursively as follows:

\[
\begin{prooftree}
\hypo{}
\infer1[$\mathsf{id}$]{A \vdash A}
\end{prooftree}
\quad \mapsto \quad
\begin{prooftree}
\hypo{}
\infer1[$\mathsf{lem}$]{\neg A, A}
\end{prooftree}
\]

\[
\begin{prooftree}
\hypo{p_0}
\infer1{\Gamma, A \vdash \Delta}
\infer1[$\neg L$]{\Gamma \vdash \neg A, \Delta}
\end{prooftree}
\quad \mapsto  \quad
\begin{prooftree}
\hypo{\tau(p_0)}
\infer1{\neg \Gamma, \neg A, \Delta}
\end{prooftree}
\]
\[
\begin{prooftree}
\hypo{p_0}
\infer1{\Gamma \vdash A, \Delta}
\infer1[$\neg_R$]{\Gamma,\neg A \vdash\Delta}
\end{prooftree}
\quad \mapsto  \quad
\begin{prooftree}
\hypo{\tau(p_0)}
\infer1{\neg \Gamma, A, \Delta}
\infer1[$\neg\neg$]{\neg \Gamma, \neg \neg A, \Delta}
\end{prooftree}
\]

\[
\begin{prooftree}
\hypo{p_0}
\infer1{\Gamma \vdash A[t/x],\Delta}
\infer1[$\exists_R$]{\Gamma \vdash\exists x A, \Delta}
\end{prooftree}
\quad \mapsto  \quad
\begin{prooftree}
\hypo{\tau(p_0)}
\infer1{\neg \Gamma,  A[t/x], \Delta}
\infer1[$\exists$]{\neg \Gamma, \exists x A, \Delta}
\end{prooftree}
\]

\[
\begin{prooftree}
\hypo{p_0}
\infer1{\Gamma, A[\alpha/x] \vdash\Delta}
\infer1[$\exists_L$]{\Gamma, \exists x A \vdash\Delta}
\end{prooftree}
\quad \mapsto  \quad
\begin{prooftree}
\hypo{\tau(p_0)}
\infer1{\neg \Gamma, \neg A[\alpha/x], \Delta}
\infer1[$\neg \exists$]{\neg \Gamma, \neg \exists x A, \Delta}
\end{prooftree}
\]

\[
\begin{prooftree}
\hypo{p_0}
\infer1{\Gamma \vdash A_i,\Delta}
\infer1[$\vee_R$]{\Gamma \vdash A_1 \vee A_2, \Delta}
\end{prooftree}
\quad \mapsto  \quad
\begin{prooftree}
\hypo{\tau(p_0)}
\infer1{\neg \Gamma,  A_i, \Delta}
\infer1[$\vee$]{\neg \Gamma, A_1 \vee A_2, \Delta}
\end{prooftree}
\]

\[
\begin{prooftree}
\hypo{p_1}
\infer1{\Gamma, A_1 \vdash \Delta}
\hypo{p_2}
\infer1{\Gamma, A_2 \vdash \Delta}
\infer2[$\vee_L$]{\Gamma, A_1 \vee A_2 \vdash\Delta}
\end{prooftree}
\quad \mapsto  \quad
\begin{prooftree}
\hypo{\tau(p_1)}
\infer1{\neg \Gamma, \neg A_1, \Delta}
\hypo{\tau(p_2)}
\infer1{ \neg \Gamma, \neg A_2, \Delta}
\infer2[$\neg \vee$]{\neg \Gamma, \neg \Gamma, \neg(A_1 \vee A_2), \Delta}
\infer1[$\mathsf{c}^*$]{\neg \Gamma, \neg(A_1 \vee A_2), \Delta}
\end{prooftree}
\]

\[
\begin{prooftree}
\hypo{p_0}
\infer1{\Gamma \vdash A,A,\Delta}
\infer1[$\mathsf{c}_R$]{\Gamma \vdash A,\Delta}
\end{prooftree}
\quad \mapsto \quad
\begin{prooftree}
\hypo{\tau(p_0)}
\infer1{\neg \Gamma, A,A, \Delta}
\infer1[$\mathsf{c}$]{\neg \Gamma, A ,\Delta}
\end{prooftree}
\]
\[
\begin{prooftree}
\hypo{p_0}
\infer1{\Gamma, A, A \vdash \Delta}
\infer1[$\mathsf{c}_L$]{\Gamma,A\vdash \Delta}
\end{prooftree}
\quad \mapsto \quad
\begin{prooftree}
\hypo{\tau(p_0)}
\infer1{\neg\Gamma,\neg A,\neg A, \Delta}
\infer1[$\mathsf{c}$]{\neg \Gamma, \neg A, \Delta}
\end{prooftree}
\]
\[
\begin{prooftree}
\hypo{p_0}
\infer1{\Gamma \vdash \Delta}
\infer1[$\mathsf{w}_R$]{\Gamma \vdash A,\Delta}
\end{prooftree}
\quad \mapsto \quad
\begin{prooftree}
\hypo{\tau(p_0)}
\infer1{\neg \Gamma,\Delta}
\infer1[$\mathsf{w}$]{\neg \Gamma, A ,\Delta}
\end{prooftree}
\]
\[
\begin{prooftree}
\hypo{p_0}
\infer1{\Gamma \vdash \Delta}
\infer1[$\mathsf{w}_L$]{\Gamma,A\vdash \Delta}
\end{prooftree}
\quad \mapsto \quad
\begin{prooftree}
\hypo{\tau(p_0)}
\infer1{\neg\Gamma, \Delta}
\infer1[$\mathsf{w}$]{\neg \Gamma, \neg A, \Delta}
\end{prooftree}
\]
\[
\begin{prooftree}
\hypo{p_1}
\infer1{\Gamma_1 \vdash A, \Delta_1}
\hypo{p_2}
\infer1{\Gamma_2, A, \vdash \Delta_2}
\infer2[$\mathsf{cut}$]{\Gamma_1,\Gamma_2 \vdash \Delta_1,\Delta_2}
\end{prooftree}
\quad\mapsto 
\quad
\begin{prooftree}
\hypo{\tau(p_1)}
\infer1{\neg \Gamma_1, A, \Delta_1}
\hypo{\tau(p_2)}
\infer1{\neg \Gamma_2, \neg A, \Delta_2}
\infer2[$\mathsf{cut}$]{\neg \Gamma_1,\neg \Gamma_2, \Delta_1,\Delta_2}
\end{prooftree}
\]

We will assume without loss of generality that all proofs are regular, i.e. that each eigen-variable $\alpha$ is introduded by exactly one occurrence of the $\neg \exists$-rule. Given a proof $p$ and an eigenvariable $\alpha$, we denote by $p[t/\alpha]$ the result of substituting everywhere in $p$ the variable $\alpha$ by $t$. Here, $t$ is any $\mathcal{L}^+$-term of type $\iota$.

\section{Formulas as types, proofs as terms}

\subsection{Formulas as types}

For each formula $A$, the associated \emph{evidence type} $[A]$ and \emph{counter-evidence type} $\langle A \rangle$ are defined as follows:
\\\\
- $[A] = \langle A \rangle = \Box$ for $A$ atomic.
\\\\
- $[\exists x A] = \iota \times  \langle \neg A \rangle$
\\\\
- $\langle \exists x A \rangle = [\exists x A] \to [\neg A]$
\\\\
- $[\neg A] = (\langle A \rangle \to [A]) \to \langle A \rangle$
\\\\
- $\langle \neg A \rangle = \langle A \rangle \to [A]$
\\\\
- $[A \vee B] =  \langle \neg A \rangle \times \langle \neg B \rangle$
\\\\
- $\langle A \vee B \rangle = (\langle \neg A \rangle \to [\neg A]) \times  (\langle \neg B \rangle \to [\neg B])$ 
\\\\
We shall write $E_A$ for an arbitrary closed inhabitant of $[A]$, which exists by Proposition \ref{p:inhabited}.  Given a formula $A$ and given closed $\mathcal{L}^+$-terms $u : [A]$ and $v : \langle A \rangle$, we define a formula $\winform{A}{u}{v}$ as follows:

\[
\begin{aligned}
\winform{R(t_1,\hdots, t_n)}{u}{v} & = R(t_1,\hdots,t_n) \\
\winform{\exists x A}{u}{v} & = \winform{A[\pi_1(u)/x]}{(\pi_2(u))((v(u))(\pi_2(u)))}{(v(u))(\pi_2(u))}\\
\winform{A \vee B}{u}{v} & = \winform{A}{(\pi_1(u))(((\pi_1(v))(\pi_1(u)))(\pi_1(u)))}{((\pi_1(v))(\pi_1(u)))(\pi_1(u))} \vee  \winform{B}{(\pi_2(u))(((\pi_2(v))(\pi_2(u)))(\pi_2(u)))}{((\pi_2(v))(\pi_2(u)))(\pi_2(u))} \\
\winform{\neg A}{u}{v} & = \neg \left(\winform{A}{v(u(v))}{u(v)} \right) 
\end{aligned}
\]

Intuitively, the formula $\winform{A}{u}{v}$ means that the evidence $u$ beats the counter-evidence $v$; or, in game-theoretic terms, that the strategy $u$ for the Verifier beats the strategy $v$ for the Falsifier. The functional interpretation of a closed sentence $A$ is:
\[\forall x \exists y  \left(\winform{A}{y}{x} \right)\]
We need to check that our definition of the formula $\winform{A}{u}{v}$ is sound, in the sense that the terms introduced on the right-hand sides of the equations are all of the appropriate types. 

In the case of the existential quantifier, we have:
\[
\begin{aligned}
u : & \;  [\exists x A ] \\
& = \iota \times \langle \neg A \rangle \\
& = \iota \times (\langle A \rangle \to [A])
\end{aligned}
\] 
and:
\[
\begin{aligned}
v : & \; \langle \exists x A \rangle \\
& = [\exists x A] \to [\neg A] \\
& = (\iota \times (\langle A \rangle \to [A])) \to [\neg A] \\
& = (\iota \times (\langle A \rangle \to [A])) \to ((\langle A \rangle \to [A]) \to \langle A \rangle)
\end{aligned}
\] 
It follows that $\pi_1(u) : \iota$, $\pi_2(u) : \langle A \rangle \to [A]$ and $v(u) : (\langle A \rangle \to [A]) \to \langle A \rangle$. So $(v(u))(\pi_2(u)) : \langle A \rangle$ and therefore $(\pi_2(u))((v(u))(\pi_2(u))) : [A]$, as required. 

For the case of disjunction, we have:
\[
\begin{aligned}
u : & \; [A \vee B] \\
& = \langle \neg A \rangle \times \langle \neg B \rangle \\
\end{aligned}
\]
and:
\[
\begin{aligned}
v: & \; \langle A \vee B \rangle  \\
& = (\langle \neg A \rangle \to [\neg A]) \times  (\langle \neg B \rangle \to [\neg B])\\
\end{aligned}
\]
It follows that: 
\[
\begin{aligned}
(\pi_1(v))(\pi_1(u)) : & \; [\neg A] \\
& = (\langle A \rangle \to [A]) \to \langle A\rangle \\
& = \langle \neg A \rangle \to \langle A \rangle
\end{aligned}
\]
We have $\pi_1(u) : \langle \neg A \rangle$ so we get 
\[((\pi_1(v))(\pi_1(u)))(\pi_1(u)) : \langle A \rangle \] as required. Furthermore, as $\pi_1(u) : \langle \neg A \rangle = \langle A \rangle \to [A]$ we also get:
\[(\pi_1(u))(((\pi_1(v))(\pi_1(u)))(\pi_1(u))) : [A]\]
as required. Similarly, we get \[((\pi_2(v))(\pi_2(u)))(\pi_2(u)) : \langle B \rangle\] and \[(\pi_2(u))(((\pi_2(v))(\pi_2(u)))(\pi_2(u))) : [B]\]

Finally, for the case of negation, we have $u : [\neg A] = (\langle A \rangle \to [A]) \to \langle A \rangle $ and $v : \langle \neg A \rangle = \langle A \rangle \to [A]$. So $u(v) : \langle A \rangle$, and hence $v(u(v)) : [A]$ as required.

The following proposition will be used freely without mention.

\begin{prop}
If $u_0 \equiv u_1$ and $v_0 \equiv v_1$ then:
\[\qfsval \winform{A}{u_0}{v_0} \leftrightarrow \winform{A}{u_1}{v_1}\]
\end{prop}

We also note the following simple observation:
\begin{prop}
\label{p:qf}
If $A$ is quantifier-free, then $\winform{A}{u}{v} = A$.
\end{prop}

\subsection{Term extraction from proofs}

To extract witnessing terms from proofs, we introduce the concept of a \emph{term transformer}, which plays the same role here as non-terminal symbols play in Herbrand schemes.
\begin{defi}
A \emph{term transformer} of signature $(U_1,\hdots, U_n,V)$ is a function $F$ mapping  $\mathcal{L}^+$-terms $u_1,\hdots,u_n$ of type $U_1,\hdots, U_n$ to an $\mathcal{L}^+$-term of type $V$, such that $F$ satisfies the following conditions:
\begin{enumerate}
\item Each free variable of $Fu_1\dotsm u_n$ is free in one of the terms $u_1,\hdots,u_n$. In particular, if $u_1,\hdots,u_n$ are all closed terms then so is $F u_1\dotsm u_n$, and if $u_1,\hdots, u_n$ are variables then  $F u_1\dotsm u_n$ is a term in which the only free variables are among $u_1,\hdots,u_n$. 
\item If $u_i \equiv v_i$ for each $i \in \{1,\hdots,n\}$ then $Fu_1\dotsm u_n \equiv Fv_1\dotsm v_n$. 
\end{enumerate}
\end{defi}
Given a proof $p$ of the end sequent $A_1,\hdots,A_n$ we shall associate with each index $i$ a term transformer $F^p_i$ of signature $(\langle A_1 \rangle,\hdots,\langle A_n\rangle, [A_i])$.
Note that, since sequents are technically speaking multisets with no internal order, the use of indices to refer to formula occurrences in a sequent is somewhat imprecise. Rather than being overly pedantic about this, we trust that the definitions will be sufficiently clear that the reader could easily work out the tedious task of formulating them in a fully precise way.

We shall define our term transformers by structural recursion on proofs, making a case distinction on the last rule of inference used in the proof $p$.  The size of a multiset $\Gamma$ of formulas will be denoted by $\vert \Gamma \vert$.

\paragraph{Case $\mathsf{lem}$:}
The proof $p$ is:
\[
\begin{prooftree}
\hypo{}
\infer1{\neg A, A}
\end{prooftree}
\]

We define corresponding term transformers by setting:

\[F^p_2 uv = u(v)\]
\[F^p_1  uv = \lambda z. v\]

To verify that these are well-defined term transformers, we focus on checking that the terms have the appropriate types. By assumption, $u : \langle \neg A \rangle = \langle A \rangle \to [A]$ and $v : \langle A \rangle$. Hence $F^p_2 uv = u(v) : [A]$ as required. In the definition of $F^p_1 uv$ as $\lambda z.v$, we have implicitly assumed that the variable $z$ is of type $\langle A \rangle \to [A]$, and so $\lambda z.v : (\langle A \rangle \to [A]) \to \langle A \rangle = [\neg A]$ as required.

\paragraph{Case $\mathsf{w}$:}

\[
\begin{prooftree}
\hypo{p_0}
\infer1{\Gamma}
\infer1{\Gamma,A}
\end{prooftree}
\]

Recall that $E_A$ is an arbitrarily chosen inhabitant of $[A]$. We set:
\[F^p_{\vert \Gamma \vert + 1}\vec{u}v = E_A\]
For $i \leq \vert \Gamma \vert$:
\[F^p_i \vec{u}v = F^{p_0}_i  \vec{u} \]
It is trivial to check that these are well-defined term transformers. 

\paragraph{Case $\mathsf{cut}$:}

\[
\begin{prooftree}
\hypo{p_1}
\infer1{\Gamma, A}
\hypo{p_2}
\infer1{\Delta, \neg A}
\infer2{\Gamma, \Delta}
\end{prooftree}
\]
 
Set $k = \vert \Gamma \vert$ and $m = \vert \Delta \vert$.  We let $h$ abbreviate $\lambda z. F^{p_1}_{k + 1} \vec{u} z$.  For $i > k$:
\[F^p_i \vec{u}\vec{v} = F^{p_2}_{i-m} \vec{v}h\]
For $i \leq k$:
\[F^p_i  \vec{u}\vec{v} = F^{p_1}_i  \vec{u}((F^{p_2}_{k + 1}  \vec{v} h)h)\]

Again, to check that these are well-defined term transformers, we check the types. In the definition of $h$ as $\lambda z. F^{p_1}_{k + 1} \vec{u} z$ it is implicit that $z : \langle A \rangle$, hence $h : \langle A \rangle \to [A] = \langle \neg A \rangle$. Hence $F^{p_2}_{k + 1}  \vec{v} h$ has the type $[\neg A] = (\langle A \rangle \to [A]) \to \langle A \rangle$, which means that  $(F^{p_2}_{k + 1}  \vec{v} h)h$ has type $\langle A \rangle$. From these observations follows immediatly from the induction hypothesis on $p_1$ and $p_2$ that each term $F^p_i \vec{u} \vec{v}$ has the appropriate type. 

\paragraph{Case $\mathsf{c}$:}
%
%
\[
\begin{prooftree}
\hypo{p_0}
\infer1{\Gamma,  A,  A}
\infer1{\Gamma, A}
\end{prooftree}
\]

Let $k = \vert \Gamma \vert$. We set:
\[F^p_{k + 1} \vec{u}v = \caseop{W}{F^{p_0}_{k + 1} \vec{u}vv}{F^{p_0}_{k+ 2} \vec{u}vv}\]
where:
\[W := \winform{A}{F^{p_0}_{\vert \Gamma \vert + 1} \vec{u}vv}{v}\]
For $i \leq \vert \Gamma \vert$ set:
\[F^p_i \vec{u}v =  F^{p_0}_i \vec{u}vv\]
It is trivial to check that these are well-defined term transformers.

\paragraph{Case $\exists$:}

\[
\begin{prooftree}
\hypo{p_0}
\infer1{\Gamma,  A[t/x]}
\infer1{\Gamma, \exists x A}
\end{prooftree}
\]

Let $h$ abbreviate $\lambda z. F^{p_0}_{\vert \Gamma \vert + 1}\vec{u}z$. 
Set:
\[F^p_{\vert \Gamma \vert + 1}  \vec{u} v = (t, h)\]
For $i \leq \vert \Gamma \vert$:
\[F^p_i \sigma \vec{u} v = F^{p_0}_i \vec{u}((v((t, h)))h)\]

We assume that $v$ is of type $\langle \exists x A \rangle = [\exists x A] \to [\neg A]$. The variable $z$ in the term $h$ is implicitly assumed to have type $\langle A \rangle$ so that $h : \langle A \rangle \to [A] = \langle \neg A \rangle$. Hence $(t,h) : \iota \times \langle \neg A \rangle = [\exists x A]$ as required. Furthermore, it follows that 
\[
\begin{aligned}
v((t,h)) : & \; [\neg A] \\
& = (\langle A \rangle \to [A]) \to \langle A \rangle \\
& \langle \neg A \rangle \to \langle A \rangle
\end{aligned}
\]
Since $h: \langle \neg A \rangle$ it follows that $(v((t,h)))h : \langle A \rangle$ as required. 

\paragraph{Case $\neg \exists$:}

Here $p$ is of the form:

\[
\begin{prooftree}
\hypo{p_0}
\infer1{\Gamma,  \neg A[\alpha/x]}
\infer1{\Gamma, \neg \exists x A}
\end{prooftree}
\]

We define a term $h : \langle \exists x A \rangle $ by:
\[h = \lambda y. (F^{p_0}_{\vert \Gamma \vert + 1}\vec{u}(\pi_2(y)))[\pi_1(y)/\alpha]\]
We set:
\[F^p_{\vert \Gamma \vert + 1}\vec{u}v = \lambda z.h\]
For $i \leq \vert \Gamma \vert$:
\[F^p_i   \vec{u}v = (F^{p_0}_i  \vec{u}(\pi_2(v(h))))[\pi_1(v(h)))/\alpha]\]

We check the types. First, the bound variable $y$ in $h$ is assumed to be of type $[\neg \exists x A]$, so that $h : [\neg \exists x A ] \to [\neg A] = \langle \exists x A \rangle$ as we claimed. Assuming the variable $z$ has type $\langle \exists x A \rangle \to [\exists x A]$, we have \[\lambda z.h : (\langle \exists x A \rangle \to [\exists x A]) \to \langle \exists x A \rangle = [\neg \exists x A]\]
as required. Furthermore, $v$ is assumed to have the type $\langle \neg \exists x A \rangle = \langle \exists x A \rangle \to [\exists x A]$. Hence $v(h) : [\exists x A]$. Since $\exists x A = \iota \times \langle \neg A \rangle$ we have $\pi_2(v(h)) : \langle \neg A \rangle$, and $\pi_1(v(h)) : \iota$. This ensures the definitions give the appropriate types.

\paragraph{Case $\vee$:}

There are two instances of this rule. In the first case the proof $p$ is of the form:

\[
\begin{prooftree}
\hypo{p_0}
\infer1{\Gamma, A}
\infer1{\Gamma,  A \vee B}
\end{prooftree}
\]

Let $k = \vert \Gamma \vert$. Let $q$ abbreviate $(\lambda z. F^{p_0}_{k+1} \vec{u} z, \lambda z. E_ B)$, where we recall that $E_B$ is an arbitrary inhabitant of $[B]$.  Set:
\[F^{p}_{k + 1} \vec{u} v  =  q\]
For $i \leq \vert \Gamma \vert$:
\[F^{p}_{i} \vec{u} v =  F^{p_0}_i \vec{u}(((\pi_1(v))(\pi_1(q)))(\pi_1(q)))\]
We check the types. We have:
\[
\begin{aligned}
q : & \; (\langle A \rangle \to [A]) \times (\langle B \rangle \to [B]) \\
& = \langle \neg A \rangle \times \langle \neg B \rangle \\
& [A \vee B]
\end{aligned}
\]
as required. It also follows that $\pi_1(q) : \langle \neg A \rangle$.

By assumption we have: 
\[
\begin{aligned}
v : & \; \langle A \vee B \rangle \\
 & = (\langle \neg A \rangle \to [\neg A]) \times (\langle \neg B \rangle \to [\neg B])
\end{aligned}
\]
so $\pi_1(v) : \langle \neg A \rangle \to [\neg B]$. Hence:
\[
\begin{aligned}
(\pi_1(v))(\pi_1(q)) : & \; [\neg A] \\
& = (\langle A \rangle \to [A]) \to \langle A\rangle \\
& = \langle \neg A \rangle \to \langle A \rangle
\end{aligned}
\]
Since $\pi_1(q) : \langle \neg A \rangle$ we get $((\pi_1(v))(\pi_1(q)))(\pi_1(q)) : \langle A \rangle$, as required.

In the second case $p$ is:

\[
\begin{prooftree}
\hypo{p_0}
\infer1{\Gamma, B}
\infer1{\Gamma,  A \vee B}
\end{prooftree}
\]

In this case let $q$ abbreviate $(\lambda z. E_A, \lambda z. F^{p_0}_{k+1} \vec{u} z)$.  Set:
\[F^{p}_{k + 1} \vec{u} v  =  q\]
For $i \leq \vert \Gamma \vert$:
\[F^{p}_{i} \vec{u} v =  F^{p_0}_i \vec{u}(((\pi_2(v))(\pi_2(q)))(\pi_2(q)))\]
Checking the types is similar to the first case.

\paragraph{Case $\neg \vee$:}

\[
\begin{prooftree}
\hypo{p_1}
\infer1{\Gamma, \neg A}
\hypo{p_2}
\infer1{\Delta, \neg B}
\infer2{\Gamma, \Delta, \neg (A \vee B)}
\end{prooftree}
\]

Let $k = \vert \Gamma \vert $ and let $m = \vert \Delta \vert$. 
Let $q_1$ abbreviate $\lambda z. F^{p_1}_{k+1} \vec{u} z$ and let $q_2$ abbreviate $\lambda z. F^{p_2}_{m+1} \vec{v} z$. Let $q$ abbreviate $\lambda z. (q_1,q_2)$. 
We set:
\[F^p_{k + m +1} \vec{u}\vec{v} w =  q \]
For $k < i \leq k+m$:
\[F^p_i \vec{u} \vec{v} w =  F^{p_2}_{i - k} \vec{v} (\pi_2(w(q(w)))\]
For $i \leq k$:
\[F^p_i \vec{u} \vec{v} w =  F^{p_1}_{i} \vec{u} (\pi_1(w(q(w)))\]

We check the types. First, we have $q_1  : \langle \neg A \rangle \to [\neg A]$ and $q_2 : \langle \neg B \rangle \to [\neg B]$. Hence:

\[
\begin{aligned}
(q_1,q_2) : & \;   (\langle \neg A \rangle \to [\neg A]) \times ( \langle \neg B \rangle \to [\neg B]) \\
& = \langle A \vee B \rangle 
\end{aligned}
\]
Assuming that $z$ is a variable of type $\langle A \vee B \rangle \to [A \vee B]$, we thus have:
\[
\begin{aligned}
\lambda z. (q_1,q_2) : & \; (\langle A \vee B \rangle \to [A \vee B ]) \to \langle A \vee B \rangle \\
& = [\neg( A \vee B)]
\end{aligned}
\]
as required. 

By assumption we have: 
\[
\begin{aligned}
w : & \;  \langle \neg (A \vee B) \rangle \\
& = \langle A \vee B \rangle \to [A \vee B] 
\end{aligned}
\]
Hence $q(w) : \langle A \vee B \rangle$, and so 
\[
\begin{aligned}
w(q(w)) : & \; [A \vee B] \\
& = \langle \neg A \rangle \times \langle \neg B \rangle
\end{aligned}
\]
so $\pi_1(w(q(w))) : \langle \neg A \rangle$ and $\pi_2(w(q(w))) : \langle \neg B \rangle$ as required. 

\paragraph{Case $\neg\neg$:}

The proof $p$ is of the form:
\[
\begin{prooftree}
\hypo{p_0}
\infer1{\Gamma, A}
\infer1{\Gamma, \neg \neg A}
\end{prooftree}
\]

Let $q$ abbreviate $\lambda y. F^{p_0}_{\vert \Gamma \vert + 1} \vec{u} y$. We set:
\[F^p_{\vert \Gamma \vert + 1}  \vec{u} v  =  \lambda z.q\]
For $i \leq \vert \Gamma \vert$:
\[F^p_i  \vec{u}v =  F^{p_0}_i  \vec{u}((v(q))q)\]

We check the types. Assuming $y$ is a variable of type $\langle A \rangle$ we have $q : \langle A \rangle \to [A] = \langle \neg A \rangle$. So assuming $z$ is of type $\langle \neg A \rangle \to [\neg A]$ we get:
\[
\begin{aligned}
\lambda z. q : & \; (\langle \neg A \rangle \to [\neg A]) \to \langle \neg A \rangle \\
& = [\neg \neg A]
\end{aligned}
\]
as required. 

By assumption we have:
\[
\begin{aligned}
v : & \; \langle \neg \neg A \rangle \\
& = \langle \neg A \rangle \to [\neg A]
\end{aligned}
\]
so $v(q) : [\neg A] = (\langle A \rangle \to [A]) \to \langle A \rangle$. Hence, as $q: \langle A \rangle \to [A]$ we get $(v(q))q : \langle A \rangle$ as required.

With the definitions in place, the following proposition can be proved by a straightforward induction:

\begin{prop}
\label{p:subst-lemma}
\[(F^p_i  \vec{u})[t/\alpha] = F^{p[t/\alpha]}_i \vec{u}\]
\end{prop}
Note that, using Proposition \ref{p:subst-lemma}, we can write the definition of terms transformers for a proof $p$ of the form
\[
\begin{prooftree}
\hypo{p_0}
\infer1{\Gamma,  \neg A[\alpha/x]}
\infer1{\Gamma, \neg \exists x A}
\end{prooftree}
\]
equivalently as follows. We let
\[h = \lambda y. F^{p_0[\pi_1(y)/\alpha]}_{\vert \Gamma \vert + 1}\vec{u}(\pi_2(y))\]
and set:
\[F^p_{\vert \Gamma \vert + 1}\vec{u}v = \lambda z.h\]
For $i \leq \vert \Gamma \vert$:
\[F^p_i   \vec{u}v = F^{p_0[ \pi_1(v(h))/\alpha]}_i  \vec{u}(\pi_2(v(h)))\]

We are now ready to state the main result:

\begin{theo}
\label{t:soundness}
Let $p$ be a proof of the sequent $A_1,\hdots,A_n$. Then, for all closed terms \[u_1 : \langle A_1\rangle, \hdots, u_n : \langle A_n \rangle\] we have:
\[\qfsval \winform{A_1}{F^p_1 \vec{u}}{u_1} \vee \hdots \vee \winform{A_n}{F^p_n \vec{u}}{u_n}\]
\end{theo}

The proof of Theorem \ref{t:soundness} will occupy us in Section \ref{s:main-proof}.

\subsection{Sequent calculus, Herbrand's theorem and concurrency}
\label{s:concurrency}

The sequent calculus has a close relationship with Herbrand's theorem; indeed, the most well-known way to prove Herbrand's theorem is probably via cut elimination for sequent calculus, extracting a Herbrand disjunction from a ``mid-sequent'' in a cut-free proof. This connection between sequent calculus and Herbrand's theorem was explored by Alcolei et al. \cite{alcolei2018true}, in which game semantics was used to describe the computational content of sequent calculus proofs.  The game takes place between two players, ``Prover'' and ``Refuter'', and allows for several copies of a game to be played in parallell. The winning condition favors Prover, in that she essentially wins the game if she wins in one of the parallel copies of the game. The main result in \cite{alcolei2018true} shows that one can extract a winning strategy for Prover for a sequent calculus proof, and that if the end sequent is $\Sigma_1$ then that strategy naturally gives rise to a Herbrand disjunction. 

To illustrate this game-theoretic perspective we can consider the following classic example, a sequent calculus proof of the formula known as the ``drinker paradox'':

\[
\begin{prooftree}
\hypo{}
\infer1{ P\alpha,  \neg P\alpha}
\infer1[$\vee$]{ P\alpha,  P\alpha \to  P\beta}
\infer1[$\neg \neg + \neg \exists$]{ P\alpha,  \forall y (P\alpha \to  Py)}
\infer1[$\exists$]{ P\alpha, \exists x \forall y (Px \to  Py)}
\infer1[$\vee$]{ Pc \to  P\alpha, \exists x \forall y (Px \to  Py)}
\infer1[$\neg \neg + \neg \exists$]{ \forall y (Pc \to  Py), \exists x \forall y (Px \to  Py)}
\infer1[$\exists$]{\exists x \forall y (Px \to  Py), \exists x \forall y (Px \to  Py)}
\infer1[$\mathsf{c}$]{\exists x \forall y (Px \to  Py)}
\end{prooftree}
\]
Here we take $\forall x : = \neg \exists x \neg$ and $A \to B := \neg A \vee B$ as usual. In the game-theoretic interpretation of this proof, Prover's main task is to come up with a witness for the existential quantifier. Howver, there is no one choice that will work for a single play of the game. So Prover immediately decides to split the game into two copies to be played in parallel, and aims to ensure a win in at least one of the copies. The way that Prover's strategy plays out is directly reflected in the two ``threads'' of the proof; in the left thread, corresponding to one copy of the game, Prover simply chooses a witness $c$ for the existential quantifier. Now, Refuter has to respond with some counter-example to the universal quantifier. The eigen-variable $\alpha$ can be thought of as a generic proposed counter-example, so that Prover's strategy does not assume anything about how Refuter is going to play. Prover has set a trap here: playing the first move $c$ forces a response from  Refuter, and this will then be used in the other copy of the game represented by the right thread in the proof. Here, Prover takes the counter-example $\alpha$ from Refuter's response and now plays that as a witness instead. In this sense there is communication between the parallel copies of the game: Prover can take a move of Refuter in one copy of the game and use it to find the right moves in a different copy of the game. All of this is plainly visible in the proof, and thinking of a sequent calculus proof as a concurrent process seems very natural indeed.

The aim of the  present work is to see how functional interpretations fit into this picture. The important thing  is not that we \emph{can} extract realizing terms from a sequent calculus proof, but rather \emph{how} this is done. Given a proof $p$ of  a sequent $A_1,\hdots,A_n$, we can think of closed terms  \[u_1 : \langle A_1\rangle, \hdots, u_n : \langle A_n \rangle\] as representing strategies for Refuter in $n$ different games played in parallel. The term extraction finds strategies for Prover to ensure that she wins in at least one of these parallel games, expressing by the validity of the disjunction:
\[ \winform{A_1}{F^p_1 \vec{u}}{u_1} \vee \hdots \vee \winform{A_n}{F^p_n \vec{u}}{u_n}\]
Thus we can think if these extracted terms as processes running concurrently. The connection between threads in a proof and individual processes is also preserved:
each individual term $F^p_i \vec{u}$ is constructed directly from the threads of the proof leading to the formula $A_i$.

Furthermore, as in the game-semantics framework, there is communication between these individual processes. To isolate the relevant part of the above proof of the drinker paradox, consider a proof $p$ of the form:
\[
\begin{prooftree}
\hypo{\vdots}
\infer1{ d.\; \neg \exists y  A(\alpha,y),   \neg A(c,\alpha)}
\infer1{c.\;\exists x \neg \exists y   A(x,y),   \neg  A(c,\alpha)}
\infer1{b.\;\exists x \neg \exists y   A(x,y),  \neg \exists y  A(c,y)}
\infer1{a.\;\exists x \neg \exists y A(x,y), \exists x \neg \exists y  A(x,y)}
\end{prooftree}
\]

Given terms $u, v$ representing counter-evidence for each copy of the formula   $\exists x \neg \exists y  A(x,y)$, we extract witnessing terms $F^a_1 u v$ and $F^a_2 u v$. The labels $a,b,c,d$ help to refer to subproofs of $p$, so that $p$ itself has the label $a$. The interesting thing is to look at  how the term $F^a_1 u v$ is extracted; the fact that this is a non-principal formula at the end-sequent of the proof can intuitively be thought of as Prover waiting for the second copy of the game to play out for a bit before making her move.  According to the definition of term transformers from proofs, we get:
\[
F^a_1 u v  = F^b_1 u ((v(c, h))h) 
\]
 where $h$ abbreviates $\lambda z. F^b_2 u z$.   So here, Prover simply observes that given the strategy played in the parallel copy of the game, where Prover plays $c$ as a witness, Refuter responds with $(v(c, h))h$ which is a term of type $\langle \neg \exists y A(c,y) \rangle$. We abbreviate this term by $w_c$, with the subscript to emphasize the dependence on the term $c$ played by Prover.  In the next step we compute:
\[
F^b_1 u w_c  = (F^c_1 u (\pi_2(w_c(g))))[\pi_1(w_c(g))/\alpha]
\]
where $g$ abbreviates $\lambda z. (F^c_2 u (\pi_2(z))[\pi_1(z)/\alpha]$. In the next step, Prover makes her move, and gives $\pi_1(w_c(g))$ as the witness for the existential quantifer due to the substitution of $\pi_1(w_c(g))$ for $\alpha$. So this move depends on a response of Refuter to a move made by Prover in another other copy of the game.

\section{Proof of the main theorem}
\label{s:main-proof}

We shall prove Theorem \ref{t:soundness} by induction on the size of a proof $p$, verifying that soundness is preserved by each recursive clause of the definition. We exclude from the proof the trivial case of a proof ending with the weakening rule. 

\subsection{Interpretation of the axiom}

Shape of $p$:
\[
\begin{prooftree}
\hypo{}
\infer1{\neg A, A}
\end{prooftree}
\]

We recall that we defined: 
\[F^p_2 uv = u(v)\]
\[F^p_1 uv = \lambda z. v\]
We have:
\[
\begin{aligned}
\winform{\neg A}{F^p_1 \sigma u v}{u} & \equiv \winform{\neg A}{\lambda z. v}{u} \\
& = \neg \left( \winform{A}{u((\lambda z. v)u)}{(\lambda z. v)u} \right) \\
& \equiv \neg \left( \winform{A}{u(v)}{v} \right) \\
\end{aligned}
\]
and:
\[
\begin{aligned}
\winform{A}{F^p_2 \sigma u v}{v} & \equiv \winform{A}{u(v)}{v}
\end{aligned}
\]
As a propositional tautology we have:
\[\qfsval  \neg \left( \winform{A}{u(v)}{v} \right) \vee \winform{A}{u(v)}{v}\]
hence:
\[\qfsval  \winform{\neg A}{F^p_1 \sigma u v}{u} \vee \winform{A}{F^p_2 \sigma u v }{v}\]
as required.

\subsection{Interpretation of cut}

Shape of $p$:

\[
\begin{prooftree}
\hypo{p_1}
\infer1{\Gamma, A}
\hypo{p_2}
\infer1{\Delta, \neg A}
\infer2{\Gamma, \Delta}
\end{prooftree}
\]

We recall the definition of the corresponding term transformers. 
Set $k = \vert \Gamma \vert$ and $m = \vert \Delta \vert$.  We let $h$ abbreviate $\lambda z. F^{p_1}_{k + 1}  \vec{u} z$ so that $h : \langle A \rangle \to [A]$.  For $i > k$:
\[F^p_i  \vec{u}\vec{v} =  F^{p_2}_{i-m} \vec{v}h\]
For $i \leq k$:
\[F^p_i  \vec{u}\vec{v} =  F^{p_1}_i  \vec{u}((F^{p_2}_{k + 1} \vec{v} h)h)\]
The induction hypothesis on $p_1$ gives:
\[(1)\;{\qfsval}\bigvee_{1 \leq i \leq k} \winform{U_i}{F^{p_1}_i  \vec{u} w }{u_i} \vee \winform{A}{F^{p_1}_{k+1}  \vec{u} w }{w} \]
for all $w$. Substituting $(F^{p_2}_{m + 1} \vec{v} h)h $ for $w$ in $(1)$ gives:
\[(2)\;{\qfsval}\bigvee_{1 \leq i \leq k} \winform{U_i}{F^{p_1}_i  \vec{u} ((F^{p_2}_{m + 1} \vec{v} h)h) }{u_i} \vee \winform{A}{F^{p_1}_{k+1}  \vec{u} ((F^{p_2}_{m + 1} \vec{v} h)h) }{(F^{p_2}_{m + 1} \vec{v} h)h} \]
The induction hypothesis on $p_2$ gives:
\[(3)\;{\qfsval}\bigvee_{1 \leq j \leq m} \winform{V_j}{F^{p_2}_j  \vec{v} w }{v_j} \vee \winform{\neg A}{F^{p_2}_{m+1}  \vec{v} w }{w} \]
For all $w$, in particular:
\[(4)\;{\qfsval}\bigvee_{1 \leq j \leq m} \winform{V_j}{F^{p_2}_j  \vec{v} h }{v_j} \vee \winform{\neg A}{F^{p_2}_{m+1}  \vec{v} h }{h} \]
But we have:
\[
\begin{aligned}
\winform{\neg A}{F^{p_2}_{m+1}  \vec{v} h }{h} & = \neg \left( \winform{A}{h((F^{p_2}_{m+1}  \vec{v} h)h)}{(F^{p_2}_{m+1}  \vec{v} h)h } \right) \\
& = \neg \left( \winform{A}{(\lambda z. F^{p_1}_{k+1}  {u} z)((F^{p_2}_{m+1}  \vec{v} h)h)}{(F^{p_2}_{m+1}  \vec{v} h)h } \right) \\
& \equiv \neg \left( \winform{A}{ F^{p_1}_{k+1}  {u}((F^{p_2}_{m+1}  \vec{v} h)h)}{(F^{p_2}_{m+1}  \vec{v} h)h } \right) \\
\end{aligned}
\]
So from $(4)$ we get:
\[(5)\;{\qfsval}\bigvee_{1 \leq j \leq m} \winform{V_j}{F^{p_2}_j  \vec{v} h }{v_j} \vee \neg \left( \winform{A}{ F^{p_1}_{k+1}  {u}((F^{p_2}_{m+1}  \vec{v} h)h)}{(F^{p_2}_{m+1}  \vec{v} h)h } \right) \]
A cut on $(5)$ and $(2)$ gives:
\[(6)\;{\qfsval}\bigvee_{1 \leq i \leq k} \winform{U_i}{F^{p_1}_i  \vec{u} ((F^{p_2}_{m + 1} \vec{v} h)h) }{u_i} \vee \bigvee_{1 \leq j \leq m} \winform{V_j}{F^{p_2}_j  \vec{v} h }{v_j}  \]
But this is:
\[{\qfsval}\bigvee_{1 \leq i \leq k} \winform{U_i}{F^{p}_i  \vec{u} \vec{v}}{u_i} \vee \bigvee_{k + 1 \leq j \leq k + m} \winform{V_j}{F^{p}_{j}  \vec{u}\vec{v}}{v_j}  \]
as required.

\subsection{Contraction}
Shape of $p$:
\[
\begin{prooftree}
\hypo{p_0}
\infer1{\Gamma,  A,  A}
\infer1{\Gamma, A}
\end{prooftree}
\]

We recall the definition of the corresponding term transformers:
\[F^p_{\vert \Gamma \vert + 1} \vec{u}v = \caseop{W}{F^{p_0}_{\vert \Gamma \vert + 1} \vec{u}vv}{F^{p_0}_{\vert \Gamma \vert + 2} \vec{u}vv}\]
where:
\[W = \winform{A}{F^{p_0}_{\vert \Gamma \vert + 1} \vec{u}vv}{v}\]
For $i \leq \vert \Gamma \vert:$
\[F^p_i \vec{u}v =  F^{p_0}_i \vec{u}vv\]
The induction hypothesis on $p_0$ gives:
\[(1)\; {\qfsval}\bigvee_{1 \leq i \leq \vert \Gamma \vert} \winform{U_i}{F^{p_0}_i \vec{u}w_1 w_2}{u_i} \vee \winform{A}{F^{p_0}_{\vert \Gamma \vert + 1} \vec{u} w_1 w_2}{w_1}  \vee  \winform{A}{F^{p_0}_{\vert \Gamma \vert + 2} \vec{u} w_1 w_2}{w_2} \]
for all $w_1,w_2$. Setting $w_1 = w_2 = v$ in $(1)$ we get:
\[(2)\;{\qfsval}\bigvee_{1 \leq i \leq \vert \Gamma \vert} \winform{U_i}{F^{p_0}_i \vec{u}v v}{u_i} \vee \winform{A}{F^{p_0}_{\vert \Gamma \vert + 1} \vec{u} vv}{v}  \vee  \winform{A}{F^{p_0}_{\vert \Gamma \vert + 2} \vec{u} v v}{v} \]
We have:
\[(3)\; W {\qfsval}\winform{A}{F^{p_0}_{\vert \Gamma \vert + 1} \vec{u} vv}{v} \leftrightarrow \winform{A}{\caseop{W}{F^{p_0}_{\vert \Gamma \vert + 1} \vec{u}vv}{F^{p_0}_{\vert \Gamma \vert + 2} \vec{u}vv}}{v}\] 
and:
\[(4)\; \neg W {\qfsval}\winform{A}{F^{p_0}_{\vert \Gamma \vert + 2} \vec{u} vv}{v} \leftrightarrow \winform{A}{\caseop{W}{F^{p_0}_{\vert \Gamma \vert + 1} \vec{u}vv}{F^{p_0}_{\vert \Gamma \vert + 2} \vec{u}vv}}{v}\] 
From $(2)$, $(3)$ and $(4)$ we get:
\[(5)\;{\qfsval}\bigvee_{1 \leq i \leq \vert \Gamma \vert} \winform{U_i}{F^{p_0}_i \vec{u}v v}{u_i} \vee  \winform{A}{\caseop{W}{F^{p_0}_{\vert \Gamma \vert + 1} \vec{u}vv}{F^{p_0}_{\vert \Gamma \vert + 2} \vec{u}vv}}{v} \]
But this is:
\[{\qfsval}\bigvee_{1 \leq i \leq \vert \Gamma \vert} \winform{U_i}{F^p_i \vec{u}v}{u_i} \vee \winform{A}{F^p_i \vec{u} v}{v}\]
as required.

\subsection{$\exists$-rule}

Shape of $p$:
\[
\begin{prooftree}
\hypo{p_0}
\infer1{\Gamma,  A[t/x]}
\infer1{\Gamma, \exists x A}
\end{prooftree}
\]

We recall the definition of the corresponding term transformers. 
Let $h$ abbreviate $\lambda z. F^{p_0}_{\vert \Gamma \vert + 1} \vec{u}z v$. 
We have:
\[F^p_{\vert \Gamma \vert + 1} \vec{u} v =  (t, h)\]
For $i \leq \vert \Gamma \vert$:
\[F^p_i \vec{u} v =  F^{p_0}_i \vec{u}((v((t, h)))h)\]
We get: 
\[
\begin{aligned}
 \winform{\exists x A}{F^p_{\vert \Gamma \vert + 1} \vec{u} v }{v}  
& \equiv \winform{\exists x A}{(t, h)}{v} \\
& \equiv \winform{A[\pi_1((t, h))/x]}{(\pi_2((t, h)))((v((t, h)))(\pi_2((t, h))))}{(v((t, h)))(\pi_2((t, h)))} \\
& \equiv \winform{A[t/x]}{h((v((t, h)))h)}{(v((t, h)))h} \\
& \equiv \winform{A[t/x]}{ F^{p_0}_{\vert \Gamma \vert + 1} \vec{u}((v((t,h)))h))}{(v((t, h)))h} \\
\end{aligned}
\]

Hence we have:
\[(1)\; \winform{\exists x A}{F^p_{\vert \Gamma \vert + 1} \vec{u} v }{v}  \equiv  \winform{A[t/x]}{ F^{p_0}_{\vert \Gamma \vert + 1} \vec{u}((v((t,h)))h))}{(v((t, h)))h} \]

The induction hypothesis on $p_0$ gives:
\[(2)\;{\qfsval}\bigvee_{1 \leq i \leq \vert \Gamma \vert} \winform{U_i}{F^{p_0}_i \vec{u} w}{u_i} \vee \winform{A[t/x]}{F^{p_0}_{\vert \Gamma \vert + 1}\vec{u}w}{w}\]
for all $w$. Setting $w = (v((t, h)))h$ in $(1)$ gives:
\[(3)\;{\qfsval}\bigvee_{1 \leq i \leq \vert \Gamma \vert} \winform{U_i}{F^{p_0}_i \vec{u} ((v((t, h)))h) v}{u_i} \vee \winform{A[t/x]}{F^{p_0}_{\vert \Gamma \vert + 1}\vec{u}((v((t, h)))h)}{(v((t, h)))h}\]
But by $(1)$ this is:
\[{\qfsval}\bigvee_{1 \leq i \leq \vert \Gamma \vert} \winform{U_i}{F^{p}_i \vec{u} v}{u_i} \vee \winform{\exists x A}{F^{p}_{\vert \Gamma \vert + 1}\vec{u}v}{v} \]
as required.

\subsection{$\neg \exists$-rule}

Shape of $p$:
\[
\begin{prooftree}
\hypo{p_0}
\infer1{\Gamma,  \neg A[\alpha/x]}
\infer1{\Gamma, \neg \exists x A}
\end{prooftree}
\]

We recall the definition of the corresponding term transformers. We set:
\[h = \lambda y. F^{p_0[\pi_1(y)/\alpha]}_{\vert \Gamma \vert + 1} \vec{u}(\pi_2(y))\]
Appealing to Proposition \ref{p:subst-lemma} the associated term transformers are:
\[F^p_{\vert \Gamma \vert + 1} \vec{u}v = \lambda z.h\]
For $i \leq \vert \Gamma \vert$:
\[F^p_i  \vec{u}v = F^{p_0 [\pi_1(v(h))/\alpha]}_i \vec{u}(\pi_2(v(h)))\]
We have:
\[
\begin{aligned}
\winform{\neg \exists x A}{F^p_{\vert \Gamma \vert + 1} \vec{u}v}{v} & = \winform{\neg \exists x A}{\lambda z. h}{v} \\
& = \neg \left( \winform{\exists x A}{v((\lambda z.h)v)}{(\lambda z.h)v} \right) \\
& \equiv \neg \left( \winform{\exists x A}{v(h)}{h} \right) \\
& = \neg \left( \winform{ A [\pi_1(v(h))/x]}{(\pi_2(v(h)))((h(v(h)))(\pi_2(v(h))) )}{(h(v(h))) (\pi_2(v(h)))} \right) \\
& = \winform{ \neg A [\pi_1(v(h))/\alpha]}{h(v(h))}{\pi_2(v(h))} \\
& = \winform{ \neg A [\pi_1(v(h))/x]}{(\lambda y. F^{p_0[\pi_1(y)/\alpha]}_{\vert \Gamma \vert + 1} \vec{u}(\pi_2(y)))(v(h))}{\pi_2(v(h))} \\
& \equiv \winform{ \neg A [\pi_1(v(h))/\alpha]}{ F^{p_0[\pi_1(v(h))/x]}_{\vert \Gamma \vert + 1} \vec{u}(\pi_2(v(h)))}{\pi_2(v(h))} \\
\end{aligned}
\]

The induction hypothesis on $p_0[\pi_1(v(h))/\alpha]$ gives:
\[(1)\; \bigvee_{1 \leq i \leq \vert \Gamma \vert} \winform{U_i}{F^{p_0 [\pi_1(v(h))/\alpha]}_i \vec{u}w}{u_i} \vee \winform{ \neg A [\pi_1(v(h))/\alpha]}{ F^{p_0[\pi_1(v(h))/\alpha]}_{\vert \Gamma \vert + 1} \vec{u}w}{w} \]
for all $w$. Setting $w = \pi_2(v(h))$ in $(1)$ gives:
\[(2)\; \bigvee_{1 \leq i \leq \vert \Gamma \vert} \winform{U_i}{F^{p_0 [\pi_1(v(h))/\alpha]}_i \vec{u}(\pi_2(v(h)))}{u_i} \vee \winform{ \neg A [\pi_1(v(h))/\alpha]}{ F^{p_0[\pi_1(v(h))/\alpha]}_{\vert \Gamma \vert + 1} \vec{u}(\pi_2(v(h)))}{\pi_2(v(h))} \]
But $(2)$ is:
\[{\qfsval}\bigvee_{1 \leq i \leq \vert \Gamma \vert} \winform{U_i}{F^p_i \vec{u}v}{u_i} \vee \winform{\neg \exists x A}{F^p_{\vert \Gamma \vert + 1} \vec{u}v}{v} \]
as required.

\subsection{$\vee$-rule}

We consider only one instance of the $\vee$-rule since the two cases are similar. 

\[
\begin{prooftree}
\hypo{p_0}
\infer1{\Gamma, A}
\infer1{\Gamma,  A \vee B}
\end{prooftree}
\]

We recall the definition of the corresponding term transformers. 
Let $k = \vert \Gamma \vert$. Let $q$ abbreviate $(\lambda z. F^{p_0}_{k+1} \vec{u} z, \lambda z. E_ B)$.  Set:
\[F^{p}_{k + 1} \vec{u} v = q\]

For $i \leq \vert \Gamma \vert$:
\[F^{p}_{i} \vec{u} v =  F^{p_0}_i \vec{u}(((\pi_1(v))(\pi_1(q)))(\pi_1(q)))\]
We have:
\[
\begin{aligned}
\winform{A \vee B}{F^{p}_{k + 1} \vec{u} v}{v} & = \winform{A \vee B}{q}{v}  \\
& = \winform{A}{\pi_1(q)(((\pi_1(v))(\pi_1(q)))(\pi_1(q)))}{((\pi_1(v))(\pi_1(q)))(\pi_1(q))} \vee  \winform{B}{\pi_2(q)(((\pi_2(v))(\pi_2(q)))(\pi_2(q)))}{((\pi_2(v))(\pi_2(q)))(\pi_2(q))} \\
& \equiv \winform{A}{(\lambda z. F^{p_0}_{k+1} \vec{u} z)(((\pi_1(v))(\pi_1(q)))(\pi_1(q)))}{((\pi_1(v))(\pi_1(q)))(\pi_1(q))} \vee  \winform{B}{(\lambda z. E_ B)(((\pi_2(v))(\pi_2(q)))(\pi_2(q)))}{((\pi_2(v))(\pi_2(q)))(\pi_2(q))} \\
& \equiv \winform{A}{F^{p_0}_{k+1} \vec{u} (((\pi_1(v))(\pi_1(q)))(\pi_1(q)))}{((\pi_1(v))(\pi_1(q)))(\pi_1(q))} \vee  \winform{B}{E_ B}{((\pi_2(v))(\pi_2(q)))(\pi_2(q))} \\
\end{aligned}
\]

By the induction hypothesis on $p_0$ we have:
\[(1)\; {\qfsval}\bigvee_{i \leq k} \winform{U_i}{F^{p_0}_i \vec{u}w}{u_i} \vee \winform{A}{F^{p_0}_{k+1} \vec{u}w}{w}\]
for all $w$. Setting $w = ((\pi_1(v))(\pi_1(q)))(\pi_1(q))$  gives:
\[(2)\; {\qfsval}\bigvee_{i \leq k} \winform{U_i}{F^{p_0}_i \vec{u}(((\pi_1(v))(\pi_1(q)))(\pi_1(q)))}{u_i} \vee \winform{A}{F^{p_0}_{k+1} \vec{u}(((\pi_1(v))(\pi_1(q)))(\pi_1(q)))}{((\pi_1(v))(\pi_1(q)))(\pi_1(q))}\]
But we have:
\[(3)\quad \winform{(A \vee B)}{F^p_{k+1} \vec{u}v}{v} = \winform{A}{F^{p_0}_{k+1} \vec{u}(((\pi_1(v))(\pi_1(q)))(\pi_1(q)))}{((\pi_1(v))(\pi_1(q)))(\pi_1(q))}\vee  \winform{B}{E_ B}{((\pi_2(v))(\pi_2(q)))(\pi_2(q))}\]

So from $(2)$ and propositional logic we get:
\[(4)\; {\qfsval}\bigvee_{i \leq k} \winform{U_i}{F^{p_0}_i \vec{u}(((\pi_1(v))(\pi_1(q)))(\pi_1(q)))}{u_i} \vee \winform{(A \vee B)}{F^p_{k+1} \vec{u}v}{v}\]
which is:
\[{\qfsval}\bigvee_{i \leq k}\winform{U_i}{F^p_i \vec{u}v}{u_i} \vee \winform{(A \vee B)}{F^p_{k+1} \vec{u}v}{v}\]
as required.

\subsection{$\neg \vee$-rule}
Shape of $p$:
\[
\begin{prooftree}
\hypo{p_1}
\infer1{\Gamma, \neg A}
\hypo{p_2}
\infer1{\Delta, \neg B}
\infer2{\Gamma, \Delta, \neg (A \vee B)}
\end{prooftree}
\]
We recall the definition of the corresponding term transformers. 
Let $k = \vert \Gamma \vert $ and let $m = \vert \Delta \vert$. 
Let $q_1$ abbreviate $\lambda z. F^{p_1}_{k+1} \vec{u} z$ and let $q_2$ abbreviate $\lambda z. F^{p_2}_{m+1} \vec{v} z$. 
We set:
\[F^p_{k + m +1} \vec{u}\vec{v} w = q \]
For $k < i \leq k+m$:
\[F^p_i \vec{u} \vec{v} w =  F^{p_2}_{i - k} \vec{v} (\pi_2(w(q(w)))\]
For $i \leq k$:
\[F^p_i \vec{u} \vec{v} w =  F^{p_1}_{i} \vec{u} (\pi_1(w(q(w)))\]

Note that, for all $u,v$, we have:
\[
\begin{aligned}
& \winform{\neg (A \vee B)}{u}{v} \\
& =  \neg \left(\winform{A \vee B }{v(u(v))}{u(v)} \right) \\
& =  \neg (\winform{A}{\pi_1(v(u(v)))(((\pi_1(u(v)))(\pi_1(v(u(v)))))(\pi_1(v(u(v)))))}{((\pi_1(u(v)))(\pi_1(v(u(v)))))(\pi_1(v(u(v))))}  \\
& \qquad \vee  \winform{B}{\pi_2(v(u(v)))(((\pi_2(u(v)))(\pi_2(v(u(v)))))(\pi_2(v(u(v)))))}{((\pi_2(u(v)))(\pi_2(v(u(v)))))(\pi_2(v(u(v))))}) \\
& \equiv \neg \left(\winform{A}{\pi_1(v(u(v)))(((\pi_1(u(v)))(\pi_1(v(u(v)))))(\pi_1(v(u(v)))))}{((\pi_1(u(v)))(\pi_1(v(u(v)))))(\pi_1(v(u(v))))}\right) \\
& \qquad \wedge \neg \left(  \winform{B}{\pi_2(v(u(v)))(((\pi_2(u(v)))(\pi_2(v(u(v)))))(\pi_2(v(u(v)))))}{((\pi_2(u(v)))(\pi_2(v(u(v)))))(\pi_2(v(u(v))))}\right)  \\
& =  \winform{\neg A }{(\pi_1(u(v)))(\pi_1(v(u(v))))}{\pi_1(v(u(v)))} \wedge \winform{\neg B}{(\pi_2(u(v)))(\pi_2(v(u(v))))}{\pi_2(v(u(v)))}\\
\end{aligned}
\]

Hence we have:
\[
\begin{aligned}
& \winform{\neg (A \vee B)}{F^p_{k + m +1} \vec{u}\vec{v} w}{w} \\
& =  \winform{\neg (A \vee B)}{q}{w} \\
& \equiv  \winform{\neg A }{(\pi_1(q(w)))(\pi_1(w(q(w))))}{\pi_1(w(q(w)))} \wedge \winform{\neg B}{(\pi_2(q(w)))(\pi_2(w(q(w))))}{\pi_2(w(q(w)))}\\
& \equiv  \winform{\neg A }{(\pi_1((\lambda z.(q_1,q_2))(w)))(\pi_1(w(q(w))))}{\pi_1(w(q(w)))} \wedge \winform{\neg B}{(\pi_2((\lambda z.(q_1,q_2))(w)))(\pi_2(w(q(w))))}{\pi_2(w(q(w)))}\\
& \equiv  \winform{\neg A }{q_1(\pi_1(w(q(w))))}{\pi_1(w(q(w)))} \wedge \winform{\neg B}{q_2(\pi_2(w(q(w))))}{\pi_2(w(q(w)))}\\
& \equiv  \winform{\neg A }{F^{p_1}_{k+1} \vec{u}(\pi_1(w(q(w))))}{\pi_1(w(q(w)))} \wedge \winform{\neg B}{F^{p_2}_{m+1} \vec{v}(\pi_2(w(q(w))))}{\pi_2(w(q(w)))}\\
\end{aligned}
\]
The induction hypothesis on $p_1$ gives:
\[(1)\; {\qfsval}\bigvee_{i \leq k} \winform{U_i}{ F^{p_1}_{i} \vec{u}r}{u_i} \vee \winform{\neg A }{F^{p_1}_{k+1} \vec{u}r}{r} \]
for all $r$. Setting $r = \pi_1(w(q(w)))$ gives:
\[(2)\; {\qfsval}\bigvee_{i \leq k} \winform{U_i}{ F^{p_1}_{i} \vec{u}(\pi_1(w(q(w))))}{u_i} \vee \winform{\neg A }{F^{p_1}_{k+1} \vec{u}(\pi_1(w(q(w))))}{\pi_1(w(q(w)))} \]
The induction hypothesis on $p_2$ gives:
\[(3)\; {\qfsval}\bigvee_{k < i \leq k + m} \winform{V_{i}}{ F^{p_2}_{i - k} \vec{v}r}{v_i} \vee \winform{\neg B }{F^{p_2}_{m+1} \vec{v}r}{r} \]
for all $r$. Setting $r = \pi_2(w(q(w)))$ gives:
\[(4)\; {\qfsval}\bigvee_{k < i \leq k + m} \winform{V_{i}}{ F^{p_2}_{i - k} \vec{v}(\pi_2(w(q(w))))}{v_i} \vee \winform{\neg B }{F^{p_2}_{m+1} \vec{v}(\pi_2(w(q(w))))}{\pi_2(w(q(w)))} \]
From $(2)$ and $(4)$ we get:
\[(5)\quad\quad
\begin{aligned} & {\qfsval}\bigvee_{i \leq k} \winform{U_i}{ F^{p_1}_{i} \vec{u}(\pi_1(w(q(w))))}{u_i}  \\
& \vee \bigvee_{k < i \leq k + m} \winform{V_{i}}{ F^{p_2}_{i - k} \vec{v}(\pi_2(w(q(w))))}{v_i} \\
&   \vee \left(  \winform{\neg A }{(\pi_1(u(v)))(\pi_1(v(u(v))))}{\pi_1(v(u(v)))} \wedge \winform{\neg B}{(\pi_2(u(v)))(\pi_2(v(u(v))))}{\pi_2(v(u(v)))}\right)
\end{aligned}
\]
But this is: 
\[{\qfsval}\bigvee_{i \leq k} \winform{U_i}{F^p_i \vec{u}\vec{v} w}{u_i} \vee \bigvee_{k < i \leq k + m} \winform{V_i}{F^p_i \vec{u}\vec{v} w}{v_i} \vee \winform{\neg (A \vee B)}{F^p_{k + m + 1} \vec{u} \vec{v} w}{w}\]
as required. 
\subsection{Double negation}
Shape of $p$:
\[
\begin{prooftree}
\hypo{p_0}
\infer1{\Gamma, A}
\infer1{\Gamma, \neg \neg A}
\end{prooftree}
\]
We recall the definition of the corresponding term transformers. 
Let $f$ abbreviate $\lambda y. F^{p_0}_{\vert \Gamma \vert + 1} \vec{u} y$. We set:
\[F^p_{\vert \Gamma \vert + 1}  \vec{u} v =  \lambda z.f\]
where $z :  \langle \neg A \rangle \to [\neg A]$ and $v : \langle \neg \neg A \rangle = \langle \neg A \rangle \to [\neg A]$. 
For $i \leq \vert \Gamma \vert$:
\[F^p_i  \vec{u}v =  F^{p_0}_i  \vec{u}((v(f))f)\]

We get:
\[
\begin{aligned}
\winform{\neg\neg A}{F^p_{\vert \Gamma \vert + 1}  \vec{u} v}{v} & = \winform{\neg\neg A}{\lambda z. f}{v} \\
& =  \neg \left( \winform{\neg A}{v((\lambda z. f)v)}{(\lambda z.f)v}\right) \\
& \equiv  \neg \left( \winform{\neg A}{v(f)}{f}\right) \\
& =  \neg \neg  \left( \winform{A}{f((v(f))f)}{(v(f))f}\right) \\
& =  \neg \neg  \left( \winform{A}{(\lambda y. F^{p_0}_{\vert \Gamma \vert + 1} \vec{u} y)((v(f))f)}{(v(f))f}\right) \\
& \equiv  \neg \neg  \left( \winform{A}{ F^{p_0}_{\vert \Gamma \vert + 1} \vec{u}((v(f))f)}{(v(f))f}\right) \\
& \equiv  \winform{A}{ F^{p_0}_{\vert \Gamma \vert + 1} \vec{u}((v(f))f)}{(v(f))f} \\
\end{aligned}
\]

The induction hypothesis on $p_0$ gives:
\[{\qfsval}\bigvee_{1 \leq i \leq \vert \Gamma \vert} \winform{U_i}{F^{p_0}_i  \vec{u}w}{u_i} \vee \winform{A}{ F^{p_0}_{\vert \Gamma \vert + 1} \vec{u}w}{w} \]
for all $w$. Setting $w = (v(f))f$ gives:
\[{\qfsval}\bigvee_{1 \leq i \leq \vert \Gamma \vert} \winform{U_i}{F^{p_0}_i  \vec{u}((v(f))f)}{u_i} \vee \winform{A}{ F^{p_0}_{\vert \Gamma \vert + 1} \vec{u}((v(f))f)}{(v(f))f} \]
which is: 
\[{\qfsval}\bigvee_{1 \leq i \leq \vert \Gamma \vert} \winform{U_i}{F^p_i  \vec{u}v}{u_i} \vee \winform{A}{ F^{p}_{\vert \Gamma \vert + 1} \vec{u}v} {v} \]
as required.

This concludes the proof of Theorem \ref{t:soundness}.

\section{Herbrand's theorem}

In this section we show how to extract Herbrand disjunctions from our functional interpretation. We first show a result that is of some independent interest, namely that \emph{any} term transformer that serves as a suitable realizer of the functional interpretation of a $\Sigma_1$-formula yields a Herbrand disjunction.

\begin{prop}
\label{p:herbrand-from-tt}
Suppose there is a term transformer $F$ of signature $\langle \exists x A \rangle \to [\exists x A]$, where $A$ is quantifier free, such that for all closed terms $v$ we have:
\[\qfsval \winform{\exists x A}{Fv}{v}\]
Then there are closed $\mathcal{L}$-terms $t_1,\hdots, t_n$ such that:
\[\qfsval A[t_1/x] \vee \dotsm \vee A[t_n/x]\]
\end{prop} 
\begin{proof}
Let $v$ be the closed term $\lambda z. E_{\neg A}$ (or indeed any arbitrarily chosen closed term of type $\langle \exists x A \rangle$). By definition of a term transformer the term $F v$ is also closed. By Proposition \ref{p:strong-normalization}, the term $F v$ reduces to some term $u : [\exists x A]$ in normal form. Since $F v \equiv u$, we get
\[\qfsval \winform{\exists x A}{u}{v} \leftrightarrow \winform{\exists x A}{F v}{v}\]
Hence, by Theorem \ref{t:soundness}, we get:
\[\qfsval \winform{\exists x A}{u}{v}\]
Since $u : [\exists x A] = \iota \times \langle \neg A \rangle$ is in normal form, by Proposition \ref{p:strong-normalization} $u$ is either of the form $\caseop{B}{w}{w'}$ or of the form $(t,w)$. In the former case, we get:
\[ \winform{\exists x A}{u}{v} \qfsval \winform{\exists x A}{w}{v} \vee \winform{\exists x A}{w'}{v}\]
In the latter case, $t$ is either of the form $\caseop{B}{s}{s'}$ or is an $\mathcal{L}$-term, again by Proposition \ref{p:strong-normalization}. In the first case we get
\[ \winform{\exists x A}{(t,w)}{v} \qfsval \winform{\exists x A}{(s,w)}{v} \vee \winform{\exists x A}{(s',w)}{v}\]
In the latter case we get:
\[\qfsval \winform{\exists x A}{(t,w)}{v} \leftrightarrow A[t/x]\]
using Proposition \ref{p:qf}. 
\end{proof}

We now get:

\begin{theo}[Herbrand's Theorem]
If the formula $\exists x A$ is provable, with $A$ quantifier free, then there are closed $\mathcal{L}$-terms $t_1,\hdots, t_n$ such that:
\[\qfsval A[t_1/x] \vee \dotsm \vee A[t_n/x]\]
\end{theo}

\begin{proof}
By Theorem \ref{t:soundness}, given a proof $p$ of $\exists x A$ the term transformer $F^p_1$ satifies the conditions of Proposition \ref{p:herbrand-from-tt}. 
\end{proof}

\section{Future work}

We conclude by mentioning a few possible directions for future work:
\\\\
- Besides elucidating the connection between Herbrand schemes and functional interpretation, we claimed that the functional interpretation presented here could be viewed as a computational interpretation of the classical sequent calculus with concurrency. In the same vein we also pointed to some conceptual connections and analogies with the game-theoretic analysis of Herbrand's theorem due to Alcolei et al \cite{alcolei2018true}. This connection deserves to be explored further, in particular it would be interesting to see if the interpretation of realizing terms as strategies can be made formally precise.  
\\\\
- We handled the branching structure of traces due to contractions with the case distinction operator introduced by Gerhardy and Kohlenbach \cite{GerhardyK05}. Another approach due to Ferreira and Ferreira  \cite{ferreira2017herbrandized} is to use a so called ``Herbrandized'' functional interpretation, which allows formation of terms for finite sets. This would perhaps be a more natural fit for the sequent calculus. It does however complicate the type structure a bit, which was already quite intricate here; this is why we used the case distinction operator instead. Formulating our results using a Herbrandized functional interpretation is thus left as a task for future work. 
\\\\
- Finally, as Herbrand schemes have been developed also for cyclic proofs in \cite{herbrandcyclic}, it seems natural to try to extend the present results to cyclic proofs as well. There is also some encouraging background to build on here, as a cyclic version of Gödel's System $\mathbf{T}$ has already been explored by Das \cite{das2020circular}.

\bibliographystyle{plain}
\bibliography{my_bib}

\begin{thebibliography}{10}

\bibitem{herbrandcyclic}
Bahareh Afshari, Sebastian Enqvist, and Graham~E Leigh.
\newblock Herbrand schemes for cyclic proofs.
\newblock {\em Journal of Logic and Computation}, 35(4):exaf013, 2025.

\bibitem{afshari2025herbrand}
Bahareh Afshari, Sebastian Enqvist, and Graham~E Leigh.
\newblock Herbrand schemes for first-order logic: B. afshari et al.
\newblock {\em Archive for Mathematical Logic}, 64(7):1007--1076, 2025.

\bibitem{APAL20}
Bahareh Afshari, Stefan Hetzl, and Graham~E. Leigh.
\newblock Herbrand's theorem as higher order recursion.
\newblock {\em Ann. Pure Appl. Log.}, 171(6):102792, 2020.

\bibitem{alcolei2018true}
Aurore Alcolei, Pierre Clairambault, Martin Hyland, and Glynn Winskel.
\newblock {The True Concurrency of Herbrand's Theorem}.
\newblock In Dan Ghica and Achim Jung, editors, {\em 27th EACSL Annual
  Conference on Computer Science Logic (CSL 2018)}, volume 119 of {\em Leibniz
  International Proceedings in Informatics (LIPIcs)}, pages 5:1--5:22,
  Dagstuhl, Germany, 2018. Schloss Dagstuhl--Leibniz-Zentrum fuer Informatik.

\bibitem{aschieri2019expansion}
Federico Aschieri, Stefan Hetzl, and Daniel Weller.
\newblock Expansion trees with cut.
\newblock {\em Mathematical Structures in Computer Science}, 29(8):1009--1029,
  2019.

\bibitem{das2020circular}
Anupam Das.
\newblock A circular version of {G}\"{o}del's {T} and its abstraction
  complexity.
\newblock {\em arXiv preprint arXiv:2012.14421}, 2020.

\bibitem{de1991dialectica}
Valeria Correa~Vaz De~Paiva.
\newblock The dialectica categories.
\newblock Technical report, University of Cambridge, Computer Laboratory, 1991.

\bibitem{ferreira2017herbrandized}
Fernando Ferreira and Gilda Ferreira.
\newblock A herbrandized functional interpretation of classical first-order
  logic.
\newblock {\em Archive for Mathematical Logic}, 56(5):523--539, 2017.

\bibitem{GerhardyK05}
Philipp Gerhardy and Ulrich Kohlenbach.
\newblock Extracting herbrand disjunctions by functional interpretation.
\newblock {\em Arch. Math. Log.}, 44(5):633--644, 2005.

\bibitem{heijltjes2010classical}
Willem Heijltjes.
\newblock Classical proof forestry.
\newblock {\em Annals of Pure and Applied Logic}, 161(11):1346--1366, 2010.

\bibitem{mckinley2013proof}
Richard McKinley.
\newblock Proof nets for herbrand’s theorem.
\newblock {\em ACM Transactions on Computational Logic (TOCL)}, 14(1):1--31,
  2013.

\bibitem{miller1987compact}
Dale~A Miller.
\newblock A compact representation of proofs.
\newblock {\em Studia Logica}, 46(4):347--370, 1987.

\bibitem{pedrot2014functional}
Pierre-Marie P{\'e}drot.
\newblock A functional functional interpretation.
\newblock In {\em Proceedings of the Joint Meeting of the Twenty-Third EACSL
  Annual Conference on Computer Science Logic (CSL) and the Twenty-Ninth Annual
  ACM/IEEE Symposium on Logic in Computer Science (LICS)}, pages 1--10, 2014.

\bibitem{troelstra2000basic}
A.~S. Troelstra and H.~Schwichtenberg.
\newblock {\em Basic Proof Theory}.
\newblock Cambridge Tracts in Theoretical Computer Science. Cambridge
  University Press, 2 edition, 2000.

\bibitem{van2012functional}
Benno van~den Berg, Eyvind Briseid, and Pavol Safarik.
\newblock A functional interpretation for nonstandard arithmetic.
\newblock {\em Annals of Pure and Applied Logic}, 163(12):1962--1994, 2012.

\end{thebibliography}

\end{document}